\documentclass[conference]{IEEEtran}

\usepackage[utf8]{inputenc}
\usepackage{verbatim}
\usepackage{optidef}

\usepackage{amsmath,amssymb,amsthm,amsfonts,mathrsfs}
\usepackage{tablefootnote}
\usepackage[margin=0.6in]{geometry}
\usepackage{hyperref} 
\usepackage{graphicx,xcolor,cases} 
\usepackage{xfrac}  
\usepackage{xargs}
  
\definecolor{DarkGreen}{rgb}{0.1,0.5,0.1}
\definecolor{DarkRed}{rgb}{0.5,0.1,0.1}
\definecolor{DarkBlue}{rgb}{0.1,0.1,0.5}

\newcommand{\blue}[1]{\textcolor{blue}{#1}}  
\newcommand{\white}[1]{\textcolor{white}{#1}}  

\def\>{\rangle} 
\def\<{\langle}

\def\Pr{{\rm Pr} } 
 \def\case#1{{\left\{  
	\begin{array}{ll}
  #1
	\end{array}
 \right.   }}

\def\Pr{{\rm Pr} }

\newtheorem{definitionenv}{Definition}
\newtheorem{lemmaenv}[definitionenv]{Lemma}
\newtheorem{theoremenv}[definitionenv]{Theorem}
\newtheorem{corollaryenv}[definitionenv]{Corollary}
\newtheorem{propositionenv}[definitionenv]{Proposition}
\newtheorem{factenv}[definitionenv]{Fact}
\newtheorem{conjectureenv}[definitionenv]{Conjecture}
\newtheorem{exampleenv}{Example}
\newtheorem{remarkenv}[definitionenv]{Remark}
\newtheorem{app-lemmaenv}[section]{Lemma}

\newenvironment{remark}{\begin{remarkenv}\rm}{\end{remarkenv}}
\newenvironment{lemma}{\begin{lemmaenv}\rm}{\end{lemmaenv}}
\newenvironment{theorem}{\begin{theoremenv}\rm}{\end{theoremenv}}

\newenvironment{app-lemma}{\begin{app-lemmaenv}\rm}{\end{app-lemmaenv}}

 \makeatletter
\def\thickhline{%
  \noalign{\ifnum0=`}\fi\hrule \@height \thickarrayrulewidth \futurelet
   \reserved@a\@xthickhline}
\def\@xthickhline{\ifx\reserved@a\thickhline
               \vskip\doublerulesep
               \vskip-\thickarrayrulewidth
             \fi
      \ifnum0=`{\fi}}
\makeatother

\newlength{\thickarrayrulewidth}
\usepackage{xspace}

\newcommand{\Tr}{\ensuremath{\mathrm{Tr}}}

\def\For{\blue{For}\xspace} 
\def\EndFor{\blue{EndFor}\xspace} 
\def\If{\blue{If}\xspace}

\def\function{{\blue{Function}}\xspace} 

\def\algofirst{{MeasQbits}\xspace}
\def\algosecond{{FindNbs}\xspace}

\def\wt{{\rm wt}}

\def\Pr{{\rm Pr} } 
 \def\case#1{{\left\{  
	\begin{array}{ll}
  #1
	\end{array}
 \right.   }}

\begin{document}
\title{Learning quantum graph states with product measurements}
\author{%
  \IEEEauthorblockN{Yingkai Ouyang}
  \IEEEauthorblockA{
Department of Electrical and Computer Engineering,\\
Centre of Quantum Technologies,\\
   National University of Singapore\\
   Email:  oyingkai@gmail.com}
\and
  \IEEEauthorblockN{Marco Tomamichel}
  \IEEEauthorblockA{
Department of Electrical and Computer Engineering,\\
Centre of Quantum Technologies\\
   National University of Singapore\\}
}

\maketitle

\begin{abstract}
We consider the problem of learning $N$ identical copies of an unknown $n$-qubit
quantum graph state with product measurements. 
These graph states have corresponding graphs where every vertex has exactly $d$ neighboring vertices.
Here, we detail an explicit algorithm that uses product measurements on multiple identical copies of such graph states to learn them.
When $n \gg d$ and $N = O(d \log(1/\epsilon)  + d^2 \log n ),$
this algorithm correctly learns the graph state with probability at least $1- \epsilon$.
From channel coding theory, we find that for arbitrary joint measurements on graph states, any learning algorithm achieving this accuracy requires at least $\Omega(\log (1/\epsilon) + d \log n)$ copies when $d=o(\sqrt n)$.
We also supply bounds on $N$ when every graph state encounters identical and independent depolarizing errors on each qubit.
\end{abstract}

\date{\today} 
\maketitle
 
\section{Introduction}
Learning of quantum states has been investigated in a multitude of settings.
In the most traditional setting, quantum tomography \cite{nielsen-chuang} studies this learning problem,
and this topic still attracts plenty of attention \cite{aaronson2007learnability,aaronson2019shadow,aaronson2019online,huang2020predicting}.
In quantum tomography, we learn a description of the quantum state to a prescribed degree of accuracy given multiple identical copies. 
Studies in quantum tomography are concerned with obtaining bounds on the minimum number $N$ of such copies for different families of quantum states.
Suppose that $\rho = |{\psi}\rangle\!\langle \psi |$ is an $n$-qubit pure state, and that the estimate of $\rho$ given by $\hat \rho$ is close to $\rho$ (for some constant precision, for example in trace distance) with probability at least $1-\epsilon$. 
For pure states, \cite[Sec. IIA]{kueng2017low} and \cite{haah2017sample} showed that for any measurement strategy, even when entangling operations are applied across multiple copies of $\rho$, we have $N = \tilde{\Theta} (2^{n} + \log \frac{1}{\epsilon} )$, omitting terms linear in $n$.
This implies that the optimal learning strategy for determining an arbitrary $n$-qubit pure state requires $N$ to be exponential in $n$.

Imposing additional structure (apart from purity) on quantum states allows substantial reduction of the number of copies $N$ required to learn $\rho$.
For instance, learning an unknown stabilizer state from the set of all stabilizer states using collective measurements can be achieved with $N$ linear in $n$ \cite{aaronson_pirsa,zhao2016fast,montanaro2017learning}. 
Bounds on $N$ for the learning of subsets of stabilizer states \cite{montanaro2020quantum} or quantum states that are stabilizer pseudomixtures \cite{lai2021learning} have also recently been obtained.

Learning a quantum state using measurements that act on a single or multiple copies 
of an $n$-qubit state $|\psi\>$ can be challenging to implement.
This is because the entangling operations across multiple qubits that these measurements require can be difficult to implement in an error-free way in practice. 
It would be highly beneficial if we could learn quantum states by simply performing product measurements.
A simple way to perform product measurements on $N$ copies of $|\psi\>$ is to measure each qubit of $|\psi\>$ either in the computation basis $B_0 = \{|0\>,  |1\>\}$ or in the Hadamard basis $B_1 = \{|+\> , |-\>\}$,
where $|\pm\> =\frac{|0\> \pm |1\>}{\sqrt 2}$ and $|0\>$ and $|1\>$ form an orthonormal basis of a qubit. 
If we measure a qubit in the basis $B_0$, we denote the measurement outcomes corresponding to $|0\>$ and $|1\>$ to be 0 and 1 respectively.
If we measure a qubit in the basis $B_1$, we denote the measurement outcomes corresponding to $|+\>$ and $|-\>$ to be 0 and 1 respectively.
While product measurements are easy to describe, 
it is unclear how many copies $N$ of $|\psi\>$ a learning algorithm that uses product measurements would need.

In this paper, we consider learning a family of graph states using product measurements. Graph states \cite{schlingemann2001quantum,dur2003multiparticle,raussendorf2003measurement,hein2004multiparty} are quantum states that correspond directly to undirected simple graphs, 
and the set of all graph states is equivalent under local Clifford operations to stabilizer states \cite{schlingemann2002stabilizer,van2005local,zeng2007local}. 
Recently Montanaro and Shao evaluated an upper bound on $N$ the number of copies of $|\psi\>$ required to learn the underlying graph of degree $d$,
and found that $N = O( d \log n )$ when collective measurements are performed on pairs of graph states \cite{montanaro2020quantum}. 

In this paper, we detail an explicit algorithm that uses product measurements to learn $n$-qubit graph states of degree $d$ when $n \gg d$. Product measurements are much simpler to perform than independent measurements; while independent measurements act on a single copy of the $n$-qubit state $|\psi\>$, product measurements measure every qubit in every $|\psi\>$ individually.
For this algorithm to correctly learn what the graph state is with probability at least $1- \epsilon$, it suffices to require that the number of copies $N$ satisfies
\begin{align}
N \ge 4ed \log(n/\epsilon) + 4ed^2 \log(ne),
\end{align}
where $\log$ denotes the natural logarithm and $e$ is Euler's constant.
From channel coding theory, we find that for arbitrary measurements on the graph states, we must have 
\begin{align}
N \ge d \log_4(nd),
\end{align}
thereby also establishing the asymptotic optimality of Montanaro and Shao's scheme.
We also supply bounds on $N$ when using our algorithm and also for any algorithm when every graph state $|\psi\>$ encounters a qubit depolarizing channel that acts independently and identically on each of its qubits.

\section{Preliminaries}

Let $G_{n,d}$ denote the set of graphs with $n$ vertices, and with vertex degree equal to $d$ \cite{godsil2001algebraic}. 
These graphs are called $d$-regular graphs on $n$ vertices. Here, we require that $d=o(\sqrt n)$. For any graph $g \in G_{n,d}$, we use $V=\{1,\dots, n\}$ denote the set of vertices, and $E(g)$ to denote the edge set of $g$.
Given any vertex $v$, let $N_g(v) = \{ i \in V : \{v,i\} \in E(g) \} $ denote the neighbor set of $v$. 
Note that $N_g(v) \subset V \backslash \{ v\}$.
The graph state $|g\>$ is defined to be the unique stabilizer state stabilized by the stabilizer generated by the generators
\begin{align}
 W_v =  X_v \prod_{j  \in N_g(v)} Z_j,
\end{align}
 for $v \in V$,  where $X_k$ and $Z_j$ denote the Pauli $X$ operator (a bit-flip) acting on the $v$th qubit and identity everywhere else, and the Pauli operator $Z$ (a phase-flip) acting on the $j$th qubit respectively.
 This means that $|g\>$ is the unique state satisfying the equations
 \begin{align}
 W_v |g\> = |g\> , \quad \forall v \in V. 
 \end{align}
 We define a function Oracle() which is function that always returns as its output the graph state $|g\>$.

In a graph state, each qubit corresponds to a vertex in a graph. In the preparation of a graph state, each qubit is first initialized as a $|+\>$ state. 
Second, for every edge in the graph, a controlled phase gate is applied between the corresponding qubits. 
Since the controlled phase gate is invariant under swapping of the pair of qubits it acts on, the distinction between the control and target qubit is not important. 
Moreover, since all controlled phase gates are diagonal in the computation basis, they must commute. 
This implies that the order in which controlled phase gates are applied is not important.

 Given any binary vector ${\bf v}$, we let $\wt({\bf v})$ denote its Hamming weight.
 Given that an $n$-bit binary vector ${\bf v}$ we define 
 \begin{align}
{\rm Maj}({\bf v}) = \case{
1 & {\rm if\ } \wt({\bf v}) > n/2\\ 
\hat Q & {\rm if\ } \wt({\bf v}) = n/2\\ 
0 & {\rm if\ } \wt({\bf v}) < n/2\\ 
}, \label{def:maj}
 \end{align}
 where $\hat Q$ is a random variable such that $\Pr[\hat Q=j]=1/2$ for $j=0,1$.
 
 \section{Learning without noise}
 
 \subsection{Algorithm}
 \label{ssec:algo}
 
We could learn what $g$ is, by applying Algorithm \algofirst which performs product measurements on $N$ copies of an unknown graph state $|g\>$. This graph state is obtained by querying Oracle() a total of $N=mr$ times where $m$ and $r$ are positive integers.
\newline

\noindent  \function\ \algofirst\!($m,r,w)$ \newline
\noindent{\bf Input:} Positive integers $m$ and $r$ and integer $w \in \{1,\dots,n\}$.\newline
\noindent{\bf Output:} Bits $x_{k,j}, m_{t,k,j}$ for $j \in V$, $k=1,\dots,m$, and $t = 1,\dots, r$. \newline
\noindent 1. {\For $k$ = 1:$m$} \newline
\noindent 2. \quad Pick ${\bf x} = (x_{k,1},\dots,x_{k,n})$ uniformly at random \newline
\noindent \white{1.}\quad\   from all $n$-bit strings where $\wt({\bf x}) = w$\newline
\noindent 3. \quad  {\For $t$ = 1:$r$} \newline
\noindent 4. \quad  \quad  Set $|g\> \gets $Oracle()\newline
\noindent 5. \quad \quad  {\For $j$ = 1:$n$} \newline
\noindent 6. \quad \quad \quad  {Measure qubit $j$ of $|g\>$ in basis $B_{x_{k,j}}$.} \newline
\noindent 7. \quad \quad \quad  Set $m_{t,k,j}$ to be the measurement outcome.\newline
\noindent 8. \quad \quad  {\EndFor} \newline
\noindent 9. \quad  {\EndFor} \newline
\noindent 10.  {\EndFor} \newline

The output of \algofirst is encoded in the binary matrices $(X)_{k,j}$ and $(M_t)_{k,j}$, where $x_{k,j}$ and $m_{t,k,j}$ denote the matrix elements in the $k$th row and $j$th columns of $(X)_{k,j}$ and $(M_t)_{k,j}$ respectively.
Based on this information, we estimate the most likely neighbor set of each vertex $v$ using the following algorithm.
\newline

\noindent \function\ \algosecond\!$((X)_{k,j}, (M_1)_{k,j} , \dots, (M_r)_{k,j})$ \newline
\noindent {\bf Input}: Binary matrices $(X)_{k,j}$, $(M_1)_{k,j}, \dots , (M_r)_{k,j})$.
\newline
\noindent {\bf Output}: $\mathcal S_1,\dots, \mathcal S_n$, where each  $\mathcal S_v$ is a tuple with each component that are subsets of $V \backslash \{v\}$ and with $|\alpha_i|=d$.\newline
\noindent 1.  {Set $\mathcal S_v = \{ \alpha \subset V \backslash \{v\} : |\alpha|=d \}$ for $v \in V$}\newline
\noindent 2. {\For $k = 1:m$}\newline
\noindent 3. \quad {Define the index set $W = \{ j : x_{k,j} = 1 \}$}\newline
\noindent 4. \quad {\For $v \in W$}\newline
\noindent 5. \quad \quad  {\For $\alpha = \{a_1, \dots , a_{d}\} \subset V \backslash W$}\newline
\noindent 6. \quad\quad \quad {\For $t  = 1:r$}\newline
\noindent 7. \quad \quad \quad\quad  {Set $s_{t,k,v,\alpha} = 
{\rm mod}(m_{t,k,v} + \sum_{j=1}^{d} m_{t,k,a_j},2)$}.\newline
\noindent 9. \quad \quad \quad {\EndFor}\newline
\noindent 8. \quad \quad\quad {Set ${\bf s}_{k,v,\alpha} = 
(s_{1,k,v,\alpha},\dots, s_{t,k,v,\alpha} )$}.\newline
\noindent 10. \quad \quad \quad {\!\!\!\If ${\rm Maj}({\bf s}_{k,v,\alpha})=1$, delete $\alpha$ from $\mathcal S_v$.} \newline
\noindent 11. \quad \quad {\EndFor}\newline
\noindent 12. \quad {\EndFor}\newline
\noindent 13.  {\EndFor}\newline

\noindent  \function LearnGraphState$(m,r,w)$ \newline
\noindent{\bf Input:} Positive integers $m$ and $r$ and integer $w \in \{1,\dots,n\}$.\newline
\noindent {\bf Output}: $(\eta_1 , \dots, \eta_n)$.\newline
\noindent 1.  {Set $( (X)_{k,j}, (M_1)_{k,j},\dots, (M_r)_{k,j} ) =$ MeasQbits$(m,r,w)$}\newline
\noindent 2. {Set $(\mathcal S_1, \dots, \mathcal S_n)=$ \algosecond$((X)_{k,j}, (M_1)_{k,j},\dots, (M_r)_{k,j} )$}\newline
\newline

The algorithm LearnGraphState succeeds if for all $v=1,\dots, n$, we have $\mathcal S_v = \{ N_g(v) \}$.
This algorithm also uses $N = mr$ copies of $|g\>$ to learn what $g$ is.

\subsection{Analysis}
Here, we analyze the algorithms introduced in Section \ref{ssec:algo} assuming that we obtain noiseless copies of $|g\>$ from querying Oracle().

Lemma \ref{lem:alpha-sampling} gives the probability that for a random $v \in V$ belongs to a fixed $d$-set $\alpha$ does not contain $v$, and also is a subset of the complement of a random $w$-set $W$. This corresponds to the probability that a given $\alpha$ is sampled at the $k$th iteration of Algorithm FindNbs.
\begin{lemma}
\label{lem:alpha-sampling}
Now let $\alpha$ be any fixed $d$-set $\alpha$ where $\alpha \subset V$.
Let $v \in V$ be random and let $W$ be a random subset of $V$ with cardinality $w$.
Then $\Pr[\alpha \cap W  = \emptyset \wedge v \in W] = p_{\rm samp}$ where
\begin{align}
p_{\rm samp} =  \frac{ w \binom {n-d}{w} } { n \binom{n}{w}} .
\end{align}
\end{lemma}
\begin{proof}
Now $\Pr[v \in W] = w/n$ 
and
$\Pr[\alpha \cap W = \emptyset] = \binom {n-d}{w} /  \binom{n}{w} $.
The events $v \in W$ and $\alpha \cap W$ are independent because $v$ is a random variable that is independent of the non-random $\alpha$. 
Hence the joint probability is the product of their individual probabilities.
\end{proof}
The sampling probability $p_{\rm samp} $ is used later in our analysis, and we will need bounds on it.

\begin{lemma}
\label{lem:psamp}
Suppose that $d \ge 2$ and $n \ge 2d^2$. 
Then 
\begin{align}
\frac{1}{2ed} \le 
\left(\frac{1}{ed}  - \frac{1}{4ed(1-d/n)}\right)(1-d/n)
\le p_{\rm samp}  \le 
\frac{1}{ed}.
\end{align}
\end{lemma}
\begin{proof}
Let $w= \lceil\frac{n-d}{d}\rceil$.
We can see that 
$\frac{n}{d}-1 \le w \le \frac{n}{d}$.
Note that 
\begin{align}
p_{\rm samp} &= \frac w n\prod_{j=0}^{w-1} \left(1 - \frac{d}{n-j} \right)
\ge 
\frac w n \left( 1 - \frac{d}{n-w+1}  \right)^{w}
\notag\\
&\ge
\frac{n-d}{dn}\left( 1 - \frac{d}{n(1-1/d)  + 1}  \right)^{n/d}.
\end{align}
Since $n/d \ge 2$ we have 
\begin{align}
p_{\rm samp} 
&\ge 
\frac {n-d}{dn} \left( 1 - \frac{d}{n}  \right)^{n/d}
=
\frac {1}{d} \left( 1 - \frac{d}{n}  \right)^{n/d+1}.
\label{psamp-lowerbound}
\end{align}
Now, note that 
\begin{align}
(1-x)^{1/x} =  1/e - x / (2e) + O(x^2).
\end{align}
Using Taylor's theorem, for $0 < x < 1$, we can show that 
\begin{align}
(1-x)^{1/x} \ge  \frac{1}{e} - \frac{x}{2e(1-x)} .
\end{align}
Therefore 
\begin{align}
\left( 1 - \frac{d}{n}  \right)^{n/d }
&\ge 
\frac{1}{e} - \frac{d/n}{2e(1-d/n)} 
\label{exp-approx-bound}.
\end{align}
Substituting \eqref{exp-approx-bound} into \eqref{psamp-lowerbound} and using $d/n \ge 1/(2d)$ gives the first lower bound for the lemma.

For the second lower bound, note that $1-d/n \ge 1-1/(2d) \ge 3/4$, 
and from the first lower bound we get
$p_{\rm samp} \ge \frac{3}{4ed}
(1-\frac{1}{2(3/4)})
=\frac{3}{4ed}
(1-\frac{1}{4(3/4)})
=
\frac{1}{2ed}$.

For the upper bound, note that 
\begin{align}
p_{\rm samp} \le \frac{w}{n} \left(1 - \frac d n \right) ^w.
\end{align}

The upper bound is a continuous function of $w$, and its derivative is monotone decreasing on $1 \le w \le n$. The derivative is positive when $w=1$ and negative when $w =n$. Hence this upper bound is maximized in the interval $[1,n]$, and is attained when $w = -1/ \log(1-d/n)$ with optimal value 
\begin{align}
\frac{-1}{ en \log(1-d/n)}
\le \frac{1}{ed}.
\end{align}
\end{proof}

Note that when $d$ grows and $d=o(n)$, Lemma \ref{lem:psamp} implies that 
\begin{align}
p_{\rm samp }\ge \frac{1}{ed} + O(1/n).
\end{align}
Note that for a fixed value of $t,k$ and $v$, we have that 
\begin{align} 
\Pr[s_{t,k,v,\alpha} = 1] 
=\case{
0 &   \alpha = N_g(v)\\
1/2 &   \alpha \neq N_g(v)\\
}.
\end{align}
This is because of two reasons. First, if $\alpha = \{a_1, \dots, a_d \} = N_g(v)$, 
the parity of the measured bits $m_{t,k,v}, m_{t,k,a_1},\dots, m_{t}$ must be even,
which means that $s_{t,k,v,\alpha}$ is always equal to zero.
Second, if $\alpha = \{a_1, \dots, a_d \} \neq N_g(v)$, 
the parity of the measured bits $m_{t,k,v}, m_{t,k,a_1},\dots, m_{t}$ is even and odd with equal probability, which means that $s_{t,k,v,\alpha}$ is always equal to 1 with probability $1/2$.
Next, it is easy to see that 
\begin{align} 
\Pr[{\rm Maj}({\bf s}_{k,v,\alpha}) = 1 ]
=\case{
0 &   \alpha = N_g(v)\\
1/2 &   \alpha \neq N_g(v)\\
}.
\label{incorrectalpha-eliminate}
\end{align}
We now specify conditions under which LearnGraphStates fails with probability at most $\epsilon$.
\begin{theorem}\label{thm:noiseless}
Let the conditions of Lemma \ref{lem:psamp} hold, and that $w= \lceil\frac{n-d}{d}\rceil$. Suppose that 
\begin{align}
m \ge  4e d \log(n/\epsilon)  + 4e d^2 \log(ne/d).
\end{align}
Then, using $N=mr$ copies of $|g\>$, the probability that Algorithm LearnGraphState gives the correct output is at least $1-\epsilon$.  
\end{theorem}
\begin{remark}
To use LearnGraphState in the noiseless setting, we can set $r=1$, so that 
$N \ge 4e d \log(n/\epsilon)  + 4e d^2 \log(ne/d)$ copies of $|g\>$ suffices to learn $|g\>$ with probability at least $1-\epsilon$.
\end{remark}
\begin{proof}[Proof of Theorem \ref{thm:noiseless}]
It suffices to show that the probability that LearnGraphState finds some $v \in V$ for which $\mathcal S_v \neq \{N_g(v)\}$ is at most $\epsilon$.

For any $v \in V$, suppose that an $\alpha \neq N_g(v)$ has been sampled $s$ times by LearnGraphState.
Using \eqref{incorrectalpha-eliminate}, the probability that a random $\mathcal S_v$ contains $\alpha$ is $2^{-s}$.
The number $s$ ranges from 0 to $m$. 
Hence, at the conclusion of LearnGraphState, 
\begin{align}
&\Pr[\alpha \in \mathcal S_v ]
\notag\\
=& \sum_{s=0}^m \binom m  s p_{\rm samp}^s  (1 - p_{\rm samp})^{m-s} 2^{-s} \notag\\
=& \sum_{s=0}^m \binom m  s (p_{\rm samp}/2)^s  (1 - p_{\rm samp})^{m-s}  \notag\\
 =& (p_{\rm samp}/2  + 1-p_{\rm samp})^m = (1-p_{\rm samp}/2)^m. 
\end{align}
Applying the union bound on all of the vertices $v$, and on all $d$-sets that are subsets of 
$V \backslash \{v\}$, we must have 
\begin{align}
n \binom {n-1}d (1-p_{\rm samp}/2)^m \le \epsilon.
\end{align}
The above inequality is equivalent to 
\begin{align}
m \log  \left( \frac{1}{1-p_{\rm samp}/2} \right) \ge \log \left( \frac{n \binom {n-1}d}{\epsilon} \right). \label{eq:base-ineq}
\end{align}
Now $\log  \left( \frac{1}{1-p_{\rm samp}/2} \right)
\ge  p_{\rm samp}/2$. Choosing $w= \frac{n-d}{d}$, from Lemma \ref{lem:psamp}, we get
$p_{\rm samp} \ge 1/(2ed)$.
Since we also have $\binom {n-1}{d} \le \binom n d \le \left(  ne/d\right) ^d$,
for \eqref{eq:base-ineq} to hold, 
it suffices to require the following inequality to hold
\begin{align}
m 
\ge (p_{\rm samp}/2)^{-1}
(\log(n/\epsilon)  + d \log(ne/d)). 
\label{eq:base-ineq2}
\end{align}
\end{proof}

\section{Learning with noise}
Now consider errors that afflict our quantum graph state. We model noise using the the qubit depolarizing channel $\mathcal D_p$, which applies the identity operator on a qubit with probability $1-p$, and with probability $p/3$, applies a bit-flip $X$, a phase flip $Z$, or both a bit-flip and a phase-flip $Y=iXZ$.
On every $|g\>$ obtained from each query of Oracle(), the $n$-qubit depolarizing channel $\mathcal D_p^{\otimes n}$ acts on $|g\>$ before the product measurements are performed.

Depolarizing noise affects measurement outcomes in the bases $B_0$ and $B_1$ in a simple way. If we measure a qubit with density matrix $\tau$ in the basis $B_0$,
the probability of obtaining the outcome $|0\>$ and $|1\>$ is 
$\Tr(\frac{I+Z}{2} \tau )$ and $\Tr(\frac{I-Z}{2} \tau )$ respectively. 
Since $\frac{I + Z}{2} Z = \frac{I + Z}{2}$, a $Z$ error does not change measurement outcomes in the basis $B_0$. 
Similarly an $X$ does not change the measurement outcome in the basis $B_1$. 
On the other hand, an $X$ error flips the measurement outcome in the $B_0$ basis, a $Z$ error flips the measurement outcome in the $B_1$ basis, and a $Y$ error flips the measurement outcome in both bases. 
Hence, for both bases $B_0$ and $B_1$, the probability that a qubit measurement outcome is flipped is $2p/3$. 

Here, Lemma \ref{lem:noisy-false-rejection} gives the probabilities that the majority function in FindNbs evaluates to 1.
\begin{lemma} \label{lem:noisy-false-rejection}
Suppose that for some real $p$ where $0 \le p < 3/4$, errors modeled by $\mathcal D_p^{\otimes n}$ occur on every oracle output $|g\>$.
Then for a fixed $v$ and $k$, 
\begin{align} 
\Pr[{\rm Maj}({\bf s}_{k,v,\alpha}) = 1 ]
=\case{
\eta &, \alpha = N_g(v) \\
1/2 &, \alpha \neq N_g(v) \\
},
\label{eq:eta}
\end{align}
where
$\eta = \sum_{t > n/2}\binom r t \left(\frac{1}{2} - \gamma\right)^t 
\left(\frac{1}{2} + \gamma\right)^{r-t}$
and
\begin{align}
\gamma = \frac{(1-4p/3)^{d+1}}{2}. \label{def:gamma}
\end{align}
\end{lemma}
\begin{proof}
The second result of Lemma \ref{lem:noisy-false-rejection} follows directly from \eqref{incorrectalpha-eliminate},
because the presence of depolarizing errors does not affect the probability of the measurement outcomes in both bases $B_0$ and $B_1$ when $\alpha \neq N_g(v)$. 
Hence only the first result of the lemma is non-trivial. 

Let $q = 2p/3$.
We now prove \eqref{eq:eta}.
Using the theory of generating functions, 
for any $t = 1, \dots, r$, we have that 
\begin{align}
&\Pr[s_{t,k,v,\alpha} = 1 |  \alpha = N_g(v)]  \notag\\
=&\frac{(1-q+q)^{d+1} - (1-q-q)^{d+1}}{2}  \notag\\
=& \frac{1}{2} - \frac{(1-2q)^{d+1}}{2} = \frac{1}{2} - \gamma.
\end{align}
The result follows from the independence of the random variables 
$s_{1,k,v,\alpha}, \dots, s_{t,k,v,\alpha}$ and definition of the majority function in \eqref{def:maj}.
\end{proof}
\begin{theorem}\label{thm:noisy}
Let $p,\epsilon \in \mathbb R$ such that $0 \le p \le 3/4$ and $\epsilon > 0$.
Suppose that every oracle evaluation in LearnGraphState returns a noisy graph state $\mathcal D_p^{\otimes n}(|g\>\<g|)$.
Then there exists some value of $N$ for which LearnGraphState learns $|g\>$ correctly with probability at least $1-\epsilon$ using $N$ copies of $|g\>$ where
\begin{align}
N =
O\left(
\frac{1-4\gamma^2}{\gamma^2}
f(\epsilon, n,d)g(\epsilon, n,d)
\right),
\end{align}  
where 
$f(\epsilon, n,d) =  \log(\epsilon^{-1})  + \log(d \log n) $,
$g(\epsilon,n,d) =  d \log\left(\epsilon^{-1}\right)  + d^2 \log n$,
and $\gamma$ is as given in \eqref{def:gamma}.
\end{theorem}
\begin{proof}[Proof of Theorem \ref{thm:noisy}]
The first part of the proof is to find a lower bound on $m$ for which the probability that some $\alpha  \neq  N_g(v)$ belongs to some $\mathcal S_v$ is at most $\epsilon/2$. This lower bound 
\begin{align}
m \ge  4e d \log\left(n/(2\epsilon)\right)  + 4e d^2 \log(ne/d)
\label{eq:mlower-noisy}
\end{align}
 can be obtained directly from Theorem \ref{thm:noiseless}.

The second part of the proof is to find an upper bound on $m$ for which the probability that $\alpha \notin \mathcal S_v$ is at most $\epsilon/2$ if $\alpha  = N_g(v)$.
Now, if $\alpha  = N_g(v)$ is sampled $s$ times, the probability that $\alpha$ is eliminated from $\mathcal S_v$ is $1-(1-\eta)^{s}$.

The overall probability that 
$\alpha$ is eliminated from $\mathcal S_v$ is
\begin{align}
&\sum_{s=0}^m
\binom m s
p_{\rm samp}^s (1-p_{\rm samp})^{m-s}(1-(1-\eta)^s) \notag\\
=&
1- ( p_{\rm samp}(1-\eta) + 1-p_{\rm samp}  )^m\notag\\
=&
1- (1 - \eta p_{\rm samp}  )^m.
\end{align}
We require that
$1- (1 - \eta p_{\rm samp}  )^m \le \frac{ \epsilon }{2}$.
This is equivalent to requiring that
\begin{align}
 (1 - \eta p_{\rm samp}  )^m \ge 1- \frac{\epsilon}{2}.
 \label{ineq:noisy1}
\end{align}
For \eqref{ineq:noisy1} to hold, since $e^{-\epsilon /2}  \ge 1-\epsilon/2$ for $0 \le \epsilon \le 1$, it suffices to require that 
\begin{align}
 (1 - \eta p_{\rm samp}  )^m \ge e^{- \epsilon/2}.
 \label{ineq:noisy2}
\end{align}
Taking the logarithm on both sides of \eqref{ineq:noisy2} shows that \eqref{ineq:noisy2} is equivalent to 
\begin{align}
m \le
 \frac{-\epsilon/2}{\log(1- \eta p_{\rm samp})}
\label{ineq:noisy3}. 
\end{align} 
Since 
\begin{align}
-\log (1-\eta p_{\rm samp} )
\le \sum_{j\ge 1 }^\infty  \eta^j p_{\rm samp}^j 
= \eta p_{\rm samp}/ (1-\eta p_{\rm samp})
\end{align}
whenever $\eta p_{\rm samp} < 1$, 
for \eqref{ineq:noisy3} to hold, it suffices to require that 
\begin{align}
m \le \frac{\epsilon ( 1 - \eta p_{\rm samp} ) }{2\eta p_{\rm samp}}
\label{ineq:noisy4}. 
\end{align} 
For \eqref{ineq:noisy4} to hold, using the upper bound $p_{\rm samp} \le 1/(ed)$ from Lemma \ref{lem:psamp}, it suffices to require that
\begin{align}
m \le \frac{\epsilon ed ( 1 - \eta/(ed) ) }{2 \eta }
\label{ineq:noisy5}. 
\end{align} 
Using the additive form of the Chernoff bound on $\eta$, we get
\begin{align}
\eta \le \exp\left({-\gamma^2 r / ( 1- 4\gamma^2)}\right).
\label{eq:eta-chernoff}
\end{align}
The trivial upper bound $\eta \le 1$ and the Chernoff upper bound on $\eta$ in \eqref{eq:eta-chernoff} imply that for \eqref{ineq:noisy5} to hold, it suffices to require that 
\begin{align}
m \le \frac{\epsilon}{2}(ed - 1 ) \exp\left({\gamma^2 r / ( 1- 4\gamma^2)}\right)
\label{ineq:noisy6}. 
\end{align} 
For the upper bound \eqref{ineq:noisy6} on $m$ to be larger than the lower bound \eqref{eq:mlower-noisy} on $m$, 
it suffices to require that 
\begin{align}
r = 
\left\lceil
\frac{1-4\gamma^2}{\gamma^2}\log\left(\frac{8ed \log(n/(2\epsilon) ) + 8ed^2 \log(ne/d)}{\epsilon(ed-1)}\right)
\right\rceil,
\end{align}
and can set $r$ to be slightly larger to ensure that the conditions in Lemma \ref{lem:psamp} hold.
Evaluating the corresponding lower bound on $N = rm$ asymptotically then gives the result.
\end{proof}

\section{Converse bound}

In this section, assuming $d=o(\sqrt n)$, we derive a lower bound on the minimum number of copies $N$ of noisy $n$-qubit graph states we require to learn it. 
We prove Theorem \ref{thm:converse} by connecting the problem of learning a graph state with that of transmitting classical information over a quantum channel, and counting the asymptotic number of graphs in $G_{n,d}$.

\begin{theorem}\label{thm:converse}
Let every query to Oracle() return a noisy graph state $\mathcal D_p^{\otimes n}(|g\>\<g|)$, where $0\le p < 3/4$. 
Let $0 \le \epsilon  < 1$. 
Then as $n$ becomes large, the minimum number of copies of $|g\>$ required to learn $|g\>$ with correctness at least $1-\epsilon$ satisfies the inequality
\begin{align}
N \ge \frac{ d \log_4(nd)
}{  (1- H(2p/3))/(1-\epsilon) + 1/n} ,
\end{align}
where $H$ denotes the binary entropy function. 
\end{theorem}
\begin{proof}
The total number of qubits used to transmit information about our graph state is $nN$.
The number of bits needed to describe a $d$-regular graph is $|G_{n,d}|$.
Hence we can transmit $\log_2 |G_{n,d}|$ bits of classical information using $nN$ qubits, which corresponds to a transmission rate of 
\begin{align}
R = \frac{\log_2 | G_{n,d} |}{nN}.
\end{align}
From King's result \cite{king2003capacity}, we know that the classical capacity of the depolarizing channel $\mathcal D_p$ is
\begin{align}
1-H(2p/3).
\end{align}
From Fano's inequality, if $\epsilon < 1$, the optimal $R$ is at most 
$R \leq \frac{C}{1-\epsilon} + \frac{1}{n}$ for all $n$.
Hence we must have
\begin{align}
N \ge 
\frac{\log_2 | G_{n,d} |
 }{ n (1- H(2p/3) )/(1-\epsilon) + 1} . \label{eq:converse1}
\end{align}

From \cite{mckay1991asymptotic}, whenever $d= o(\sqrt n)$, 
we know that
\begin{align}
|G_{n,d} |=
\frac{(nd)!}{(nd/2)! 2^{nd/2} (d!)^n} \exp\left(  - \frac{d^2-1}{4} - \frac{d^3}{12n} + O(d^2/n) \right).
\end{align}
Since $\log(n!) = n \log n - n +O(\log n)$, we find that 
\begin{align}
\log |G_{n,d} |
=& nd \log (nd) - nd + O(\log(nd) ) \notag\\
&- (nd/2) \log (nd/2) + nd/2 + O(\log(nd/2))\notag\\
&-(nd/2) \log 2 - n (d \log d -d + O(\log d) )\notag\\
&- \frac{d^2-1}{4} - \frac{d^3}{12n} + O(d^2/n) .
\end{align}
Simplifying this, we see that whenever $d= o(\sqrt n)$, we have 
\begin{align}
\log |G_{n,d} |  
&= \frac{1}{2} nd \log(nd) + O(n d \log d)\notag\\
&= \frac{1}{2} nd \log(nd)\left( 1 + O\left( \frac{\log d}{\log n d}\right)\right).
\label{logGnd}
\end{align}
Substituting \eqref{logGnd} into \eqref{eq:converse1} gives the result.
\end{proof}

Note that Theorem \ref{thm:converse} implies that for constant $p$ and when $\epsilon < 1/2$, we have a lower bound of $N = \Omega(d \log n)$.

The lower bound $N = \Omega(\log(1/\epsilon)$ follows from state discrimination. Consider two graph states in the set that are are as close as possible. The probability of making a discrimination error using an optimal strategy is exponentially small (with the quantum Chernoff exponent~\cite{PhysRevLett.98.160501}). Solving for $N$ gives this dependence.

\section{Discussions}\label{sec:discussions}
Our analysis of learning graph states using product measurements leaves a number of open problems. 
Since our converse bound applies for learning algorithms that use arbitrary measurements, this leaves open the question as to whether our algorithm is asymptotically optimal for product measurements. 
For instance, one question pertaining to our lower bound for $m$ in Theorem \ref{thm:noiseless} is, whether the quadratic scaling with respect to $d$ is necessary.
For future work, it would also be interesting to know if an adaptive algorithm can asymptotically outperform our randomized algorithm.

\section{Acknowledgements}\label{sec:acknow}
Y.O. is supported by the Quantum Engineering Programme grant NRF2021-QEP2-01-P06.
M.T and Y.O are supported in part by NUS startup grants (R-263-000-E32-133 and R-263-000-E32-731). 

\bibliography{ref}{}

\begin{thebibliography}{10}

\bibitem{nielsen-chuang}
M.~A. Nielsen and I.~L. Chuang, {\em {Quantum Computation and Quantum
  Information}}.
\newblock Cambridge University Press, second~ed., 2000.

\bibitem{aaronson2007learnability}
S.~Aaronson, ``The learnability of quantum states,'' {\em Proceedings of the
  Royal Society A: Mathematical, Physical and Engineering Sciences}, vol.~463,
  no.~2088, pp.~3089--3114, 2007.

\bibitem{aaronson2019shadow}
S.~Aaronson, ``Shadow tomography of quantum states,'' {\em SIAM Journal on
  Computing}, vol.~49, no.~5, pp.~STOC18--368, 2019.

\bibitem{aaronson2019online}
S.~Aaronson, X.~Chen, E.~Hazan, S.~Kale, and A.~Nayak, ``Online learning of
  quantum states,'' {\em Journal of Statistical Mechanics: Theory and
  Experiment}, vol.~2019, no.~12, p.~124019, 2019.

\bibitem{huang2020predicting}
H.-Y. Huang, R.~Kueng, and J.~Preskill, ``Predicting many properties of a
  quantum system from very few measurements,'' {\em Nature Physics}, vol.~16,
  no.~10, pp.~1050--1057, 2020.

\bibitem{kueng2017low}
R.~Kueng, H.~Rauhut, and U.~Terstiege, ``Low rank matrix recovery from rank one
  measurements,'' {\em Applied and Computational Harmonic Analysis}, vol.~42,
  no.~1, pp.~88--116, 2017.

\bibitem{haah2017sample}
J.~Haah, A.~W. Harrow, Z.~Ji, X.~Wu, and N.~Yu, ``Sample-optimal tomography of
  quantum states,'' {\em IEEE Transactions on Information Theory}, vol.~63,
  no.~9, pp.~5628--5641, 2017.

\bibitem{aaronson_pirsa}
D.~Gottesman, ``Identifying stabilizer states,'' aug 2008.
\newblock https://pirsa.org/08080052.

\bibitem{zhao2016fast}
L.~Zhao, C.~A. P{\'e}rez-Delgado, and J.~F. Fitzsimons, ``Fast graph operations
  in quantum computation,'' {\em Physical Review A}, vol.~93, no.~3, p.~032314,
  2016.

\bibitem{montanaro2017learning}
A.~Montanaro, ``Learning stabilizer states by bell sampling,'' {\em arXiv
  preprint arXiv:1707.04012}, 2017.

\bibitem{montanaro2020quantum}
A.~Montanaro and C.~Shao, ``Quantum algorithms for learning a hidden graph and
  beyond,'' {\em arXiv preprint arXiv:2011.08611}, 2020.

\bibitem{lai2021learning}
C.-Y. Lai and H.-C. Cheng, ``Learning quantum circuits of some $ t $ gates,''
  {\em arXiv preprint arXiv:2106.12524}, 2021.

\bibitem{schlingemann2001quantum}
D.~Schlingemann and R.~F. Werner, ``Quantum error-correcting codes associated
  with graphs,'' {\em Physical Review A}, vol.~65, no.~1, p.~012308, 2001.

\bibitem{dur2003multiparticle}
W.~D{\"u}r, H.~Aschauer, and H.-J. Briegel, ``Multiparticle entanglement
  purification for graph states,'' {\em Physical review letters}, vol.~91,
  no.~10, p.~107903, 2003.

\bibitem{raussendorf2003measurement}
R.~Raussendorf, D.~E. Browne, and H.~J. Briegel, ``Measurement-based quantum
  computation on cluster states,'' {\em Physical review A}, vol.~68, no.~2,
  p.~022312, 2003.

\bibitem{hein2004multiparty}
M.~Hein, J.~Eisert, and H.~J. Briegel, ``Multiparty entanglement in graph
  states,'' {\em Physical Review A}, vol.~69, no.~6, p.~062311, 2004.

\bibitem{schlingemann2002stabilizer}
D.~Schlingemann, ``Stabilizer codes can be realized as graph codes,'' {\em
  Quantum Information \& Computation}, vol.~2, no.~4, pp.~307--323, 2002.

\bibitem{van2005local}
M.~Van~den Nest, J.~Dehaene, and B.~De~Moor, ``Local unitary versus local
  clifford equivalence of stabilizer states,'' {\em Physical Review A},
  vol.~71, no.~6, p.~062323, 2005.

\bibitem{zeng2007local}
B.~Zeng, H.~Chung, A.~W. Cross, and I.~L. Chuang, ``Local unitary versus local
  clifford equivalence of stabilizer and graph states,'' {\em Physical Review
  A}, vol.~75, no.~3, p.~032325, 2007.

\bibitem{godsil2001algebraic}
C.~Godsil and G.~F. Royle, {\em Algebraic graph theory}.
\newblock Springer, 2001.

\bibitem{king2003capacity}
C.~King, ``The capacity of the quantum depolarizing channel,'' {\em IEEE
  Transactions on Information Theory}, vol.~49, no.~1, pp.~221--229, 2003.

\bibitem{mckay1991asymptotic}
B.~D. McKay and N.~C. Wormald, ``Asymptotic enumeration by degree sequence of
  graphs with degrees $o(n^{1/2})$,'' {\em Combinatorica}, vol.~11, no.~4,
  pp.~369--382, 1991.

\bibitem{PhysRevLett.98.160501}
K.~M.~R. Audenaert, J.~Calsamiglia, R.~Mu\~noz Tapia, E.~Bagan, L.~Masanes,
  A.~Acin, and F.~Verstraete, ``Discriminating states: The quantum chernoff
  bound,'' {\em Phys. Rev. Lett.}, vol.~98, p.~160501, Apr 2007.

\end{thebibliography}
\bibliographystyle{ieeetr}

\end{document}